\documentclass[12pt]{iopart}
\usepackage{color}
\usepackage{iopams}
\usepackage{amsthm}
\usepackage[dvipdfm]{graphicx}

\newcommand{\bra}[1]{\langle {#1} |}
\newcommand{\ket}[1]{| {#1} \rangle}

\newcommand{\R}[1]{\vec{R}_{#1}}

\newtheorem{theorem}{Theorem}
\newtheorem{lemma}{Lemma}

\begin{document}
\title[Two-party LOCC convertibility and Kraus-Cirac number]{Two-party LOCC convertibility of quadpartite states and Kraus-Cirac number of two-qubit unitaries}
\author{Akihito Soeda$^1$, Seiseki Akibue$^1$ and Mio Murao$^{1,2}$}

\address{$^1$ Department of Physics, Graduate School of Science,
the University of Tokyo, 7-3-1, Hongo, Bunkyo-ku, Tokyo, Japan}
\address{$^2$ Institute for Nano Quantum Information Electronics,
the University of Tokyo, 4-6-1, Komaba, Meguro-ku, Tokyo, Japan}

\ead{soeda@phys.s.u-tokyo.ac.jp}
\ead{akibue@eve.phys.s.u-tokyo.ac.jp}
\ead{murao@phys.s.u-tokyo.ac.jp}

\begin{abstract}
Nonlocal properties (globalness) of a non-separable unitary determine how the unitary affects the entanglement properties of a quantum state.
We apply a given two-qubit unitary on a quadpartite system including two reference systems and analyze its \textit{LOCC partial invertibility} under two-party LOCC.
A decomposition given by Kraus and Cirac for two-qubit unitaries shows that the globalness is completely characterized by three parameters.
Our analysis shows that the number of non-zero parameters (the Kraus-Cirac number) has an operational significance when converting entanglement properties of multipartite states.
All two-qubit unitaries have the Kraus-Cirac number at most 3, while those with at most 1 or 2 are equivalent, up to local unitaries, to a controlled-unitary or matchgate, respectively.
The presented operational framework distinguishes the untaries with the Kraus-Cirac number 2 and 3, which was not possible by the known measure of the operator Schmidt decomposition.
We also analyze how the Kraus-Cirac number changes when two or more two-qubit unitaries are applied sequentially.
\end{abstract}
\maketitle

\section{Introduction}
In terms of entanglement, quantum information processing (QIP) is a sequence of operations that convert entanglement, in which each operation is introduced to convert a quantum state with a particular entanglement property to another state with possibly a different property.
Appropriate conversion of entanglement is necessary to achieve the desired QIP.
The feasibility of an entanglement conversion depends on the initial state, the set of available operations, and the number of times each operation is allowed to be used.

Entanglement convertibility has been investigated under various conditions.
Operations can be restricted to local operations and classical communication (LOCC) or allow \textit{global} operations, which are operations not implementable by LOCC.
Under LOCC, entanglement is viewed as a resource and led to the resource theory of entanglement (for reviews, see e.g., Refs.~\cite{VP,H}).
With global operations, entanglement conversion that increases or creates entanglement is also possible.
The rich structure of entanglement implies that the amount and the kind of entanglement that can be created depends on the given global operation.
There is a number of literature investigating the amount of maximum entanglement by a single use of a given unitary~\cite{NDDG+}, where the maximum entanglement increase (with respect to a certain entanglement measure) is obtained as a function of the parameters that describe the used unitary.

Unitaries on two-qubit systems are one of the most elementary types of global operations, parametrized by 15 degrees of freedom.
The number of relevant parameters is reduced to 3 for most typical entanglement measures~\cite{Makhlin,KC}, because the measures are defined so to be invariant under local unitaries on the individual subsystems.
Any two-qubit unitary can be expressed as
\begin{equation}
 U = u_\mathrm{A} \otimes u_\mathrm{B} \cdot \exp[\rmi (\alpha_\mathrm{x} XX + \alpha_\mathrm{y} YY + \alpha_\mathrm{z} ZZ)] \cdot v_\mathrm{A} \otimes v_\mathrm{B}, \label{KCdecomp}
\end{equation}
where $u_\mathrm{A}$, $u_\mathrm{B}$, $v_\mathrm{A}$, and $v_\mathrm{B}$ are single-qubit unitaries and $XX$, $YY$, and $ZZ$ are short-hand notations of $X \otimes X$, $Y \otimes Y$, and $Z \otimes Z$, respectively~\cite{KC,Zhang}.
As seen here, a two-qubit unitary is realizable by combining three two-qubit unitaries, i.e.\ $u_\mathrm{A} \otimes u_\mathrm{B}$, $\exp[\rmi (\alpha_\mathrm{x} XX + \alpha_\mathrm{y} YY + \alpha_\mathrm{z} ZZ)]$, and $v_\mathrm{A} \otimes v_\mathrm{B}$.
The only part not implementable within LOCC---the \textit{global} part---is determined by the three parameters, $\alpha_\mathrm{x}$, $\alpha_\mathrm{y}$, and $\alpha_\mathrm{z}$.
The generators of the global part are not unique, e.g.\ any three independent linear combinations of $XX$, $YY$, and $ZZ$ would suffice~\cite{KG}. 
If the generators are the tensor products of two local operators as in \Eref{KCdecomp}, we call it \textit{Kraus-Cirac (KC) decomposition}.
The global parameters of the KC decomposition, which we call {\it KC coefficients}, determine global properties of a two-qubit unitary in entanglement conversion.

The KC decomposition is unique if we restrict the global parameters to the Weyl chamber~\cite{Zhang}.
In this paper, we call the number of nonzero KC coefficients as \textit{KC number}.
Its operational significance in entanglement conversion was not known to the best of our knowledge.
A main contribution of this paper is to present an entanglement conversion task, whose feasibility depends on the KC number of the entangling global unitary.
In particular, we analyze entanglement convertibility of quadpartite states under a single use of a two-qubit unitary and two-party LOCC.

Let us consider two qubits, A and B, each of which may be entangled to another quantum system (or its \textit{reference system}), $\mathrm{R_A}$ and $\mathrm{R_B}$, respectively.
We assume that the reference systems are inaccessible, i.e.\ no operations are allowed on these systems.
Let the initial state be a product state with respect to partition A$\mathrm{R_A}$-B$\mathrm{R_B}$. 
We take the dimension of the reference systems arbitrary and assume that the entire system is in a pure state.
A global unitary on system A and B then generates entanglement across partition A$\mathrm{R_A}$-B$\mathrm{R_B}$.
Only LOCC on A and B are allowed after this entangling operation.
Here, the local operations include all the generalized measurements on a single-qubit, whose physical implementations are discussed in Ref.~\cite{Ota}.

The unitary converts an A$\mathrm{R_A}$-B$\mathrm{R_B}$ product state to an A$\mathrm{R_A}$-B$\mathrm{R_B}$ entangled one. 
Since all unitaries are invertible, the opposite conversion is also possible.
This requires a global operation on A and B, because the inverse of a global unitary is also global.
Therefore, if the operations on A and B are strictly limited to LOCC, the A$\mathrm{R_A}$-B$\mathrm{R_B}$ entangled state cannot convert to the initial A$\mathrm{R_A}$-B$\mathrm{R_B}$ product state.
Although LOCC allows measurements, which are probabilistic operations, we require that the inversion is deterministic since the unitary inversion is also deterministic.

On the other hand, such LOCC may still achieve a less demanding conversion.
A trivial example is to require only that the A$\mathrm{R_A}$-B$\mathrm{R_B}$ entanglement is broken. 
A much less trivial conversion is when we require breaking the A$\mathrm{R_A}$-B$\mathrm{R_B}$ entanglement by only accessing system A and B with LOCC, while in addition requiring that one or the other of the two original subsystems A$\mathrm{R_A}$ or B$\mathrm{R_B}$ be restored to its original state (Figure \ref{fig:loccpartialinversion}).
In this case, the globalness of the unitary does not automatically imply non-invertibility under LOCC.

\begin{figure}[htbp]
 \begin{center}
  \includegraphics[width=140mm]{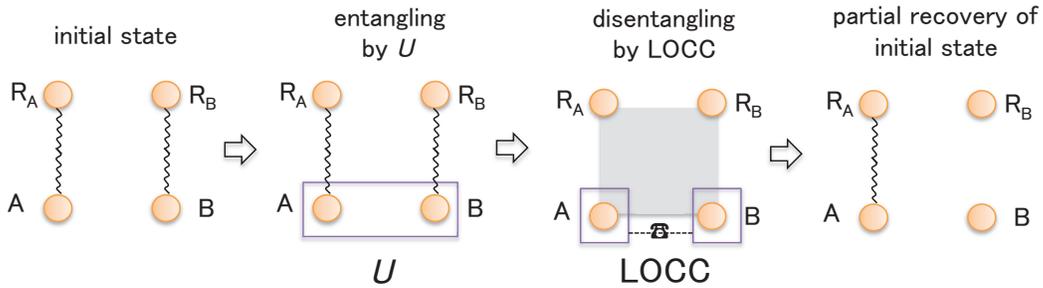}
 \end{center}
 \caption{A schematic picture of LOCC partial inversion.  The quadpartite initial state consists of two qubits A, B and their reference systems $\mathrm{R_A}$ and $\mathrm{R_B}$.  The initial state is a product state with respect to partition A$\mathrm{R_A}$-B$\mathrm{R_B}$.  The initial state is entangled by a unitary $U$ on A and B and then disentangled by two-party LOCC on A and B.   LOCC partial inversion is a task to partially recover the initial state of system A$\mathrm{R_A}$ or B$\mathrm{R_B}$ after the entangling and disentangling processes of entanglement conversion. 
 
  }
 \label{fig:loccpartialinversion}
\end{figure}

Clearly, our partial invertibility by LOCC is determined by the choice of unitary \textit{and} input state.
The difficulty of LOCC partial inversion is that we must break the entanglement across A$\mathrm{R_A}$-B$\mathrm{R_B}$ but preserve of that across A-$\mathrm{R_A}$ or B-$\mathrm{R_B}$.
Therefore, it is reasonable to expect that LOCC partial inversion is easier if A or B is initially disentangled with its respective reference system (c.f.\ A is already assumed to be disentangled with B$\mathrm{R_B}$, and B with A$\mathrm{R_A}$).
In this paper, we study which two-qubit unitary is LOCC partially invertible when (i) both A and B are initially entangled to their respective reference systems, (ii) only A or B is, and finally (iii) when neither is, and prove that each case corresponds to two-qubit unitaries with the KC number at most 1, 2, and 3, respectively.
The analysis also establishes the KC number as a figure of globalness of the unitaries.

Although two-qubit unitaries and LOCC considered are bipartite operations, they are applied on \textit{multipartite} systems.
The necessary and sufficient condition is known for entanglement convertibility of bipartite pure states under LOCC in both deterministic~\cite{majorization} and probabilistic~\cite{Vidal} conversion.
No such condition is known for multipartite states (except for 3-qubit pure states in deterministic LOCC conversion~\cite{Tajima}).
Reference~\cite{deVincente} suggests that no solution exists at least in a simple form.

Note that all two-qubit unitaries are LOCC partially invertible in the third case, because it implies that the initial state of system AB is separable and can be prepared by LOCC.
Therefore, for case (iii), the KC number can be at most 3, which is the largest possible value.
The other two cases are much less trivial.
The rest of the paper is organized as follows.
Section~\ref{LOCCinv} proves our statement regarding case (ii) and proceeds to case (i).
The KC number as a figure of the globalness of two-qubit unitaries is discussed in Section~\ref{as figure}.
We then analyze how the KC number changes when two or more two-qubit unitaries are combined in Section~\ref{moreKCs}.

\section{LOCC partial invertibility and KC number}\label{LOCCinv}
\subsection{KC number $\leq$ 2}
Let us first treat case (ii) where only system A$\mathrm{R_A}$ or B$\mathrm{R_B}$ is initially entangled, which we take the former to be entangled without loss of generality.
We omit system $\mathrm{R_B}$ in the following analysis, because the initial state of $\mathrm{R_B}$ clearly does not affect LOCC partial invertibility when B and $\mathrm{R_B}$ are disentangled.
We denote the initial state of A$\mathrm{R_A}$ by $\ket{\Phi_{\mathrm{AR_A}}}$ or $\Phi_{\mathrm{AR_A}} = \ket{\Phi_{\mathrm{AR_A}}}\bra{\Phi_{\mathrm{AR_A}}}$ and that of B  by $\ket{\varphi_{\mathrm{B}}}$ or $\varphi_{\mathrm{B}} = \ket{\varphi_{\mathrm{B}}}\bra{\varphi_{\mathrm{B}}}$.

The goal of LOCC partial inversion in this case is to find, for the entangling unitary $U$, a CPTP map $\Lambda_{\mathrm{LOCC}}$ that is implementable by LOCC and restores system A$\mathrm{R_A}$ to $\ket{\Phi_{\mathrm{AR_A}}}$, i.e.
\begin{equation}\label{PI2}
 \Tr_{\mathrm{B}}[\Lambda_{\mathrm{LOCC}}(U (\Phi_{\mathrm{AR_A}} \otimes \varphi_{\mathrm{B}}) U^\dag)] = \Phi_{\mathrm{AR_A}}.
\end{equation}
The state $\ket{\Phi_{\mathrm{AR_A}}}$ can be expressed as
\begin{equation}
 \ket{\Phi_{\mathrm{AR_A}}} = \ket{0}_\mathrm{A} \ket{v_0}_\mathrm{R_A} + \ket{1}_\mathrm{A} \ket{v_1}_\mathrm{R_A},
\end{equation}
where $\ket{0}_\mathrm{A}, \ket{1}_\mathrm{A} \in \mathcal{H}_\mathrm{A}$ are orthonormal and $\ket{v_0}_\mathrm{R_A}, \ket{v_1}_\mathrm{R_A} \in \mathcal{H}_\mathrm{R_A}$ are orthogonal but not necessarily normalized.
We define a linear operator $J: \mathcal{H}_\mathrm{R_A} \rightarrow \mathcal{H}_\mathrm{R'_A}$, where $\mathcal{H}_\mathrm{R'_A}$ is a two-dimensional system, such that
\begin{equation}
 J \ket{v_k}_\mathrm{R_A} = \ket{k}_\mathrm{R'_A} \quad (k = 0,1).
\end{equation}
If we apply $J$ from left and $J^\dag$ from right to the both sides of \Eref{PI2}, we obtain
\begin{equation}\label{PI2'}
  \Tr_{\mathrm{B}}[\Lambda_{\mathrm{LOCC}}(U (\tilde{\Phi}_{\mathrm{AR'_A}} \otimes \varphi_{\mathrm{B}}) U^\dag)] = \tilde{\Phi}_{\mathrm{AR'_A}},
\end{equation}
where $\tilde{\Phi}_{\mathrm{AR'_A}}$ is a maximally entangled (but unnormalized) state between A and $\mathrm{R'_A}$, namely
\begin{equation}
 \ket{\tilde{\Phi}_{\mathrm{AR'_A}}} = \ket{0}_\mathrm{A}\ket{0}_\mathrm{R'_A} + \ket{1}_\mathrm{A}\ket{1}_\mathrm{R'_A}.
\end{equation}
Therefore, we see that if LOCC partial inversion is possible for a given unitary $U$ for some possibly non-maximally entangled state $\ket{\Phi_{AR_A}}$, then the same unitary must be partially invertible by LOCC for a maximally entangled input.

\Eref{PI2'} implies that operations on A is limited to random unitary.
The operation $\Lambda_\mathrm{LOCC}$ and the partial trace over system B cannot increase entanglement between systems A and $\mathrm{R'_A}$.
This implies that the state $U (\tilde{\Phi}_{\mathrm{AR'_A}} \otimes \varphi_{\mathrm{B}}) U^\dag$ is maximally entangled across partition $\mathrm{AB}$-$\mathrm{R'_A}$.
Since A and $\mathrm{R'_A}$ are both two-dimensional systems, the state is also maximally entangled across partition A-$\mathrm{BR'_B}$.
Therefore, the operation on A can only be random unitary operations.
This proves the following lemma.
\begin{lemma} \label{lem:random_unitary}
If the partial inversion map $\Lambda_{\mathrm{LOCC}}$ can be implemented by LOCC for the given unitary $U$, maximally entangled input $\tilde{\Phi}_{\mathrm{AR'_A}}$, and one-qubit state $\varphi_\mathrm{B}$ as in \Eref{PI2'}, then operations on A are restricted to random unitary.
\end{lemma}

In general, $\Lambda_\mathrm{LOCC}$ is a two-way LOCC protocol, but Lemma \ref{lem:random_unitary} implies that if such a protocol exists, then one-way LOCC suffices, i.e.\
\begin{lemma} \label{one-wayLOCC}
If the partial inversion map $\Lambda_{\mathrm{LOCC}}$ can be implemented by LOCC for the given unitary $U$, maximally entangled input $\tilde{\Phi}_{\mathrm{AR'_A}}$, and one-qubit state $\varphi_\mathrm{B}$ as in \Eref{PI2'}, then $\Lambda_{\mathrm{LOCC}}$ can be implemented by one-way LOCC.
\end{lemma}

\begin{proof}
We provide a construction of a one-way LOCC protocol that implements $\Lambda_\mathrm{LOCC}$.
All quantum operations are given by a generalized measurement operation, which is specified by a set of operators $\{ M^{(r)} \}_r$ satisfying the completeness relation $\sum_r M^{(r)\dag}M^{(r)} = \mathbb{I}$.
Here, the outcomes of the measurement operation is denoted by the superscript $r$.
Any LOCC protocol on two quantum systems can be reduced to one where operations on A and B are applied in turns.
The outcome of each turn is communicated via classical communication to the other party at the end of each turn.
In general, the choice of measurement operations at each turn may depend on all the measurement outcomes of the previous turns. 

Without loss of generality, we assume that the partial inversion protocol always terminates after $n$ turns.
After $n$ turns, there is a sequence of $n$ measurement outcomes, which we denote by $\R{n}$. 
The accumulated effect of all the measurement operations is given by the product of all the operators corresponding to each measurement outcome.
Let $A^{\R{n}}$ and $B^{\R{n}}$ denote the accumulated operator for measurement sequence $\R{n}$ on system A and B, respectively.
Moreover, operator $A^{\R{n}} \otimes B^{\R{n}}$ as a whole forms a generalized measurement, i.e.\
\begin{equation} \label{ABcompleteness}
 \sum_{\R{n}} (A^{\R{n}} \otimes B^{\R{n}})^\dag A^{\R{n}} \otimes B^{\R{n}} = \mathbb{I}.
\end{equation}

By Lemma \ref{lem:random_unitary}, $A^{\R{n}}$ is proportional to a unitary, hence
\begin{equation} \label{B}
 A^{\R{n}} = c^{\R{n}} u^{\R{n}}
\end{equation}
for some positive number $c^{\R{n}}$ and unitary $u^{\R{n}}$.
Substituting \Eref{B} to \Eref{ABcompleteness} we obtain
\begin{equation}
 \sum_{\R{n}} (c^{\R{n}} B^{\R{n}})^\dag (c^{\R{n}} B^{\R{n}}) = \mathbb{I}.
\end{equation}
This implies that $\{ c^{\R{n}} B^{\R{n}} \}_{\R{n}}$ forms a generalized measurement.
We apply this measurement on B and follow by $u^{\R{n}}$ on A, which implements the desired CPTP map $\Lambda_\mathrm{LOCC}$.
\end{proof}

We are now prepared to prove our main statement:
\begin{theorem} \label{KC2}
The partial inversion map $\Lambda_{\mathrm{LOCC}}$ can be implemented by LOCC for the given unitary $U$ and maximally entangled input $\tilde{\Phi}_{\mathrm{AR'_A}}$ with some one-qubit state $\varphi_\mathrm{B}$ as in \Eref{PI2'} if and only if the KC number of the entangling unitary is at most 2.
\end{theorem}

\begin{proof}
First we prove the forward implication.
Multiply $\rho_\mathrm{R'_A}$ to the both sides of \Eref{PI2'} and take the partial trace on system $\mathrm{R'_A}$ to yield
\begin{equation} \label{PICPTP}
  \Tr_{\mathrm{B}}[\Lambda_{\mathrm{LOCC}}(U (^t\rho_\mathrm{A} \otimes \varphi_{\mathrm{B}}) U^\dag)] = {}^t\rho_\mathrm{A}.
\end{equation}
By Lemma \ref{one-wayLOCC}, we only need to consider $\Lambda_\mathrm{LOCC}$ implemented by one-way LOCC.
The unitary $U$ in \Eref{PICPTP} can be viewed as ``noise" between system A and B, while the one-way LOCC $\Lambda_\mathrm{LOCC}$ is a correction map for this noise.
This problem is precisely the one studied by Gregoratti and Werner in Ref.~\cite{GW} from which we know that such a LOCC protocol exists if and only if a CPTP map $\Gamma$ on A defined by
\begin{equation} \label{Gamma}
 \Gamma(\rho_\mathrm{A}) = \Tr_{\mathrm{B}}[U (\rho_\mathrm{A} \otimes \varphi_{\mathrm{B}}) U^\dag]
\end{equation}
is a random unitary map.
According to the quantum Birkhoff theorem~\cite{qBirkhoff}, a CPTP map on a qubit is a random unitary if and only if it is unital, i.e.\
\begin{equation}
 \Gamma(\mathbb{I}) = \mathbb{I}.
\end{equation}
This argument shows that \Eref{PICPTP} is satisfied if and only if CPTP map $\Gamma(\rho_\mathrm{A})$ for a given $U$ and $\varphi_\mathrm{B}$ is unital.

Next we analyze the unitality of $\Gamma$ in terms of the KC number.
It is easy to see that LOCC partial invertibility depends only on the global part of the KC decomposition, hence without loss of generality
we assume that in \Eref{KCdecomp}
\begin{equation}
 u_\mathrm{A} = u_\mathrm{B} = v_\mathrm{A} = v_\mathrm{B} = \mathbb{I}.
\end{equation}
We parametrize $\ket{\varphi_\mathrm{B}}$ by
\begin{equation}
 \ket{\varphi_\mathrm{B}} = a \ket{0_\mathrm{B}} + b \ket{1_\mathrm{B}},
\end{equation}
where $a$ and $b$ are normalized by $|a|^2 + |b|^2 = 1$.
Direct calculation will reveal that
\begin{equation}
 \Gamma(\mathbb{I}) = \left[ \begin{array}{cc} g_{11} & g_{12} \\ g_{21} & g_{22} \end{array} \right],
\end{equation}
where
\begin{eqnarray}
 g_{11} &= 1 + (|b|^2 - |a|^2) \sin(2 \alpha_\mathrm{x}) \sin(2 \alpha_\mathrm{y})\\
 g_{12} &= g_{21}^* = -(ab^* (\sin(2 \alpha_\mathrm{y}) - \sin(2 \alpha_\mathrm{x})) \nonumber \\
 & \quad + ba^* (\sin(2 \alpha_\mathrm{x}) + \sin(2 \alpha_\mathrm{y}))) \sin(2 \alpha_\mathrm{z})\\
 g_{22} &= 1 - (|b|^2 - |a|^2) \sin(2 \alpha_\mathrm{x}) \sin(2 \alpha_\mathrm{y}).
\end{eqnarray}

Hence, $U$ is correctable by one-way LOCC if and only if
\begin{eqnarray}
 g_{11} &= g_{22} = 1\\
 g_{12} &= g_{21}^* = 0.
\end{eqnarray}
Solving these equations, we can easily verify that the maximum number of nonzero $\alpha_k$ $(k=\mathrm{x,y,z})$ is 2.  
Therefore, LOCC partial inversion under the stated conditions is possible only if the entangling unitary has the KC number at most 2.

 The converse is almost trivial.
Notice that if the KC number is at most 2, then there exists a one-qubit state $\varphi_\mathrm{B}$ such that $\Gamma$ is a random unitary map (e.g.\ choose $\ket{\varphi_\mathrm{B}} = (\ket{0}_\mathrm{B} + \ket{1}_\mathrm{B})/\sqrt{2})$, hence a LOCC correction map exists by Gregoratti and Werner's result, and \Eref{PICPTP} is satisfied.
Finally, \Eref{PICPTP} implies \Eref{PI2'}.
\end{proof}
Thus we see that partial invertibility by LOCC under case (ii) distinguishes two-qubit unitaries with KC number 3 from other unitaries with a lower KC number.

\subsection{KC number $\leq$ 1}
Let us analyze case (i), where not only A$\mathrm{R_A}$ but B$\mathrm{R_B}$ is also initially entangled.
We denote the initial state of B$\mathrm{R_B}$ by $\ket{\Phi'_\mathrm{BR_B}}$.
LOCC partial inversion in this case is to find a LOCC-implementable map $\Lambda_\mathrm{LOCC}$ such that
\begin{equation} \label{case1}
 \Tr_\mathrm{BR_B} \left[ \Lambda_\mathrm{LOCC}(U (\Phi_\mathrm{AR_A} \otimes \Phi'_\mathrm{BR_B}) U^\dag) \right] = \Phi_\mathrm{AR_A}.
\end{equation}
Since $\Phi_\mathrm{AR_A}$ is a pure state, such $\Lambda_\mathrm{LOCC}$ exists if and only if there exists $\rho_\mathrm{BR_B}$ such that
\begin{equation} \label{case1'}
 \Lambda_\mathrm{LOCC}(U (\Phi_\mathrm{AR_A} \otimes \Phi'_\mathrm{BR_B}) U^\dag) = \Phi_\mathrm{AR_A} \otimes \rho_\mathrm{BR_B}.
\end{equation}
The state $\ket{\Phi'_{\mathrm{BR_B}}}$ can be expressed as
\begin{equation}
 \ket{\Phi'_{\mathrm{BR_B}}} = \ket{0}_\mathrm{B} \ket{w_0}_\mathrm{R_B} + \ket{1}_\mathrm{B} \ket{w_1}_\mathrm{R_B},
\end{equation}
where $\ket{w_0}_\mathrm{R_B}, \ket{w_1}_\mathrm{R_B} \in \mathcal{H}_\mathrm{R_B}$ are orthogonal but not necessarily normalized.
We define a linear operator $K: \mathcal{H}_\mathrm{R_B} \rightarrow \mathcal{H}_\mathrm{R'_B}$, where $\mathcal{H}_\mathrm{R'_B}$ is a two-dimensional system, such that
\begin{equation}
 K \ket{w_k}_\mathrm{R_B} = \ket{k}_\mathrm{R'_B} \quad (k = 0,1).
\end{equation}
We apply $J$, $J^\dag$, $K$, and $K^\dag$ to \Eref{case1'} to obtain
\begin{equation}\label{case1''}
 \Lambda_\mathrm{LOCC}(U (\tilde{\Phi}_\mathrm{AR'_A} \otimes \tilde{\Phi}_\mathrm{BR'_B}) U^\dag) = \tilde{\Phi}_\mathrm{AR'_A} \otimes \rho'_\mathrm{BR'_B},
\end{equation}
where $\rho'_\mathrm{BR'_B} = K \rho_\mathrm{BR_B} K^\dag$.
Therefore, we see that if LOCC partial inversion is possible for a given unitary $U$ for some possibly non-maximally entangled states $\ket{\Phi_{AR_A}}$ and $\ket{\Phi'_{BR_B}}$, then the same unitary must be partially invertible by LOCC for two maximally entangled inputs.
Hence, our main statement for case (i) is given as follows:
\begin{theorem} \label{KC1}
The partial inversion map $\Lambda_{\mathrm{LOCC}}$ can be implemented by LOCC for the given unitary $U$ and maximally entangled input $\tilde{\Phi}_{\mathrm{AR'_A}}$ and $\tilde{\Phi}_\mathrm{BR'_B}$ as in \Eref{case1''} with some $\rho'_\mathrm{BR'_B}$  if and only if the KC number of the entangling unitary is at most 1.
\end{theorem}

The most difficult part of the proof is that we need to treat \textit{all} the cases in which such an inversion map is found.
We are not allowed to impose any restrictions on the LOCC protocols other than that it satisfies the conditions stated in Theorem \ref{KC1}.
We shall utilize a theorem we proved in Ref.~\cite{SM}.
\begin{theorem}[adapted from Ref.~\cite{SM}] \label{NJP}
For a given two-qubit unitary $U$, there exists an LOCC-implementable map $\Lambda_\mathrm{LOCC}$ such that, for arbitrary input state $\ket{\psi_\mathrm{A}}$ and $\ket{\psi_\mathrm{B}}$,
\begin{equation} \label{LOCCreloc}
 \Lambda_\mathrm{LOCC} \left[ U (\psi_\mathrm{A} \otimes \psi_\mathrm{B})  U^\dagger \right] = \psi_\mathrm{A} \otimes \rho_\mathrm{B},
\end{equation}
where $\rho_\mathrm{B} \in \mathcal{S}(\mathcal{H}_\mathrm{B})$ is a density matrix independent of $\ket{\psi_\mathrm{A}}$, if and only if $U$ is equivalent to a controlled-unitary up to local unitaries.
\end{theorem}
An overview of the proof is: (i) first we prove that operations on A is restricted to random unitary; (ii) thus, if $\Lambda_\mathrm{LOCC}$ exists, then one-way LOCC suffices; (iii) we show that we only need to consider projective measurements for operations on B; and finally, (iv) if $\Lambda_\mathrm{LOCC}$ is implemented by a projective measurement followed by a unitary, then $U$ must be equivalent to a controlled-unitary up to local unitaries.
For more detailed proof, we refer the reader to Ref.~\cite{SM}.

\begin{proof}[Proof of Theorem \ref{KC1}]
 It is easy to see that \Eref{LOCCreloc} implies \Eref{case1''}.
For the converse, multiply ${}^t\psi_\mathrm{R'_A} \otimes {}^t\psi_\mathrm{R'_B}$ to both sides of \Eref{case1''} and take partial trace over $\mathrm{R'_AR'_B}$.
Finally, notice that
\begin{equation}
 \exp[\rmi \alpha_\mathrm{z} ZZ] = (\mathbb{I} \otimes \exp[\rmi \alpha_\mathrm{z} Z]) (\ket{0}_\mathrm{A}\bra{0} \otimes \mathbb{I} + \ket{1}_\mathrm{A}\bra{1} \otimes \exp[-\rmi 2\alpha_\mathrm{z} Z]),
\end{equation}
and that the spectral decomposition of a single-qubit unitary $u$ is given by
\begin{equation}
 u = v \exp[-\rmi \theta Z] v^\dag,
\end{equation}
for some real number $\theta$ and single-qubit unitary $v$, which show that the set of all two-qubit controlled-unitaries and their local unitary equivalents is equal the set of all two-qubit unitaries with the KC number 1.
Therefore, we proved that Theorems \ref{KC1} and \ref{NJP} are equivalent statements.
\end{proof}
Theorem \ref{KC1} shows that LOCC partial invertibility under LOCC in case (i) distinguishes unitaries with the KC number less than or equal to 1 from other unitaries.

\section{KC number as figure of globalness} \label{as figure}
The results in the previous section show that the KC number is actually a figure of globalness of the unitaries.
The entangled states which are LOCC partially invertible are, in a sense, less entangled than those not.
Two-qubit unitaries with higher KC number generate ``more" entangled states and thus have a stronger global effect on entanglement.

The degree of globalness in terms of the KC number seems to be independent of the degree in terms of other known measures of globalness.
It is not so surprising that the KC number behaves distinctly~\cite{SM} from \textit{continuous} measures of globalness such as entangling power~\cite{ZanardiZF} or entangling capacity~\cite{LHL}, but the discreteness of the KC number is not the deciding factor.
The operator Schmidt number is an operator equivalent of the Schmidt number and takes value 1, 2, and 4 for two-qubit unitaries~\cite{Nielsen}.
All global unitaries have operator Schmidt number greater than 1.
The operator Schmidt number of a given unitary is related to the Schmidt number of entangled states that can be created by the unitary.
Indeed, after a two-qubit unitary $U$ is applied in our quadpartite system, its operator Schmidt number gives the maximum Schmidt number of the resulting state with respect to partition A$\mathrm{R_A}$-B$\mathrm{R_B}$.
Unitaries with operator Schmidt number equal to 1 and 2 are equivalent to ones with KC number 0 and 1, respectively, but unitaries with higher KC number are indistinguishable by the operator Schmidt number.

Interestingly, the KC number successfully identifies a family of two-qubit unitaries known as matchgates, which produces a nontrivial class of classically simulable quantum computation~\cite{matchgate}.
All matchgates have the KC number at most 2 and, conversely, all such two-qubit unitaries are equal to a matchgate up to local unitaries.
Lastly, two-qubit unitaries with the KC number at most 1 are equivalent to a controlled-unitary up to local unitaries and vice versa.

\section{Conversion of KC number} \label{moreKCs}
A two-qubit unitary is expected to change the KC number of another two-qubit unitary when the former is applied on the latter.
We considered above when only a single two-qubit unitary is applied.
If we apply more unitaries, then the LOCC partial invertibility of the state will change.
The two lemmas below describe how much the LOCC partial invertibility can be changed.

Let us denote the set of two-qubit unitaries by $\mathcal{U}(\mathcal{H}_\mathrm{AB})$ and the set of two-qubit controlled-unitaries and their local unitary equivalents by $\mathcal{U}_c(\mathcal{H}_\mathrm{AB})$.
Note that every operation of $\mathcal{U}_\mathrm{c}(\mathcal{H})$ is local unitarily equivalent to the controlled-phase operation $C_\mathrm{p}(\theta) :=(\ket{00}\bra{00}+\ket{01}\bra{01}+\ket{10}\bra{10}+e^{i\theta}\ket{11}\bra{11})_\mathrm{AB}$.
We denote the KC number of a unitary $U$ by $\#\textsc{KC}(U)$.
 We summarize the class of $SU(4)$ in Table~\ref{table:KCnumber}.
\begin{table}[h]
\begin{center}
\begin{tabular}{|c|c|c|}
\hline
$\#\textsc{KC}(U)$ & operator Schmidt number & class\\
\hline
0 & 1 & local unitary\\
\hline
1 & 2 & controlled-unitary\\
\hline
2 & 4 & matchgate \\
\hline
3 & 4 & $SU(4)$\\
\hline
\end{tabular}
\end{center}
\caption{Classification of $SU(4)$ by the KC number and the operator Schmidt number.}
\label{table:KCnumber}
\end{table}

\begin{lemma} \label{KCcombined}
\begin{eqnarray}
\{U\in SU(4)|\#\textsc{KC}(U)\leq1\}=\{U|U\in\mathcal{U}_c(\mathbb{C}^4)\}\label{eq:controlledKC1}\\
\{U\in SU(4)|\#\textsc{KC}(U)\leq2\}=\{UV|U,V\in\mathcal{U}_c(\mathbb{C}^4)\}\label{eq:controlledKC2}\\
\{U\in SU(4)|\#\textsc{KC}(U)\leq3\}=\{UVW|U,V,W\in\mathcal{U}_c(\mathbb{C}^4)\}\label{eq:controlledKC3}
\end{eqnarray}
\end{lemma}
\begin{proof}
\Eref{eq:controlledKC1} follows from the fact that all unitaries $U$ such that $\#\textsc{KC}(U)\leq 1$ are a local unitary equivalent of the controlled-phase operations and the KC number of controlled-phase operations is 1.
The KC number of all unitary operations $U\in SU(4)$ is less than or equal to 3 and it is known that we can implement any unitary operations of $SU(4)$ by using 3 CNOT operations and local unitaries, which imply \Eref{eq:controlledKC3}.

To prove \Eref{eq:controlledKC2}, note that unitary operation $U_2=\exp{\left(-\rmi (\alpha XX + \gamma ZZ)\right)}$ is implemented by the circuit given in Figure \ref{fig:one}, which shows that $\{U\in SU(4)|\#\textsc{KC}(U)\leq2\}\subseteq\{UV|U,V\in\mathcal{U}_\mathrm{c}(\mathbb{C}^4)\}$. 
\begin{figure}[htbp]
 \begin{center}
  \includegraphics[width=60mm]{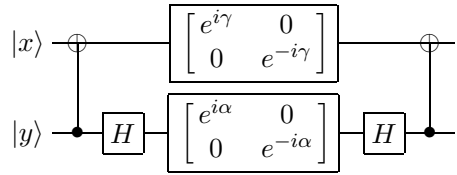}
 \end{center}
 \caption{The circuit that implements $U_2$.}
 \label{fig:one}
\end{figure}
Next, we show $\{U\in SU(4)|\#\textsc{KC}(U)\leq2\}\supseteq\{UV|U,V\in\mathcal{U}_\mathrm{c}(\mathbb{C}^4)\}$.
$U$ and $V$ are equivalent to some controlled-phase operations, which we denote as $C_\mathrm{p}(\theta_1)$ and $C_\mathrm{p}(\theta_2)$, up to local unitaries.
Hence the unitary operation $UV$ is local-unitarily equivalent to
\begin{equation}
W=C_\mathrm{p}(\theta_1)\left(R_\mathrm{y}(\phi_1)\otimes R_\mathrm{y}(\phi_2)\right)C_\mathrm{p}(\theta_2),
\end{equation}
where $R_\mathrm{y}(\phi)=\exp(-\rmi \frac{\phi}{2}Y)$.
The proof of the forward implication of Theorem \ref{KC2} also shows that \Eref{PICPTP} is satisfied if and only if CPTP map $\Gamma(\rho_\mathrm{A})$ for a given $U$ and $\varphi_\mathrm{B}$ is unital.
The unitary $W$ indeed forms a unital CPTP by setting $U \rightarrow W$ and $\varphi_\mathrm{B} \rightarrow \ket{0}_\mathrm{B}\bra{0}$ in \Eref{Gamma}.
Moreover, \Eref{PICPTP} and \Eref{PI2'} are equivalent.
Therefore, by Theorem \ref{KC2}, $\#\textsc{KC}(UV)=\#\textsc{KC}(W)\leq2$.
\end{proof}
Here, the interesting case is \Eref{eq:controlledKC2}, which shows that two two-qubit unitaries with the KC number 1 cannot be used to generate the highest level of entanglement in terms of LOCC partial invertibility.

\begin{lemma}
For any two unitaries $U,V\in SU(4)$,
\begin{equation}
|\#\textsc{KC}(U)-\#\textsc{KC}(V)|\leq \#\textsc{KC}(UV)\leq \#\textsc{KC}(U)+\#\textsc{KC}(V).
\end{equation}
\label{lemma:KC}
\end{lemma}
\begin{proof}
We assume without loss of generality that $\#\textsc{KC}(U)\geq\#\textsc{KC}(V)$.
First we prove 
\begin{equation}
\#\textsc{KC}(UV)\leq \#\textsc{KC}(U)+\#\textsc{KC}(V).
\end{equation}
The only nontrivial proof is the case 
\begin{equation}
\#\textsc{KC}(U)=\#\textsc{KC}(V)=1.
\label{eq:KC1}
\end{equation}
By \Eref{eq:KC1}, $U\in\mathcal{U}_\mathrm{c}(\mathbb{C}^4)$ and $V\in\mathcal{U}_\mathrm{c}(\mathbb{C}^4)$.
Therefore, by Lemma \ref{KCcombined}, we can show $\#\textsc{KC}(UV)\leq 2$. 

Next we prove
\begin{equation}
\#\textsc{KC}(U)-\#\textsc{KC}(V)\leq \#\textsc{KC}(UV).
\end{equation}
The nontrivial cases are
\begin{eqnarray}
\#\textsc{KC}(U)=2,\#\textsc{KC}(V)=1\label{X1}\\
\#\textsc{KC}(U)=3,\#\textsc{KC}(V)=1\label{X2}\\
\#\textsc{KC}(U)=3,\#\textsc{KC}(V)=2.\label{X3}
\end{eqnarray}
Let us first consider cases of Equations (\ref{X1}) and (\ref{X3}).
If $\#\textsc{KC}(UV)=0$, we can let $UV=u_\mathrm{A}\otimes u_\mathrm{B}$, where $u_\mathrm{A} \in \mathcal{U}(\mathcal{H}_\mathrm{A})$ and $u_\mathrm{B} \in \mathcal{U}(\mathcal{H}_\mathrm{B})$, thus
\begin{equation}
\#\textsc{KC}(U)=\#\textsc{KC}(u_\mathrm{A}\otimes u_\mathrm{B} V^{\dag})=\#\textsc{KC}(V^{\dag})=\#\textsc{KC}(V).
\end{equation}
Note that for any $SU(4)$, $\#\textsc{KC}(U)=\#\textsc{KC}(U^{\dag})$.
For case (\ref{X2}), if $\#\textsc{KC}(UV)\leq1$, we set $UV=W$.
Therefore,
\begin{equation}
\#\textsc{KC}(U)=\#\textsc{KC}(W V^{\dag})\leq\#\textsc{KC}(W)+\#\textsc{KC}(V)\leq 2.
\end{equation}

\end{proof}
The proven inequality shows that two sequentially applied unitaries with a low KC number cannot change much the LOCC invertibility of the states.

\section{Conclusion}
We studied entanglement conversion of quadpartite states under a single use of a non-separable two-qubit unitary and two-party LOCC.
The quadpartite system consists of two qubits, which may be initially entangled to its own reference systems, but otherwise disentangled.
We analyzed the LOCC partial invertibility of the system after a unitary is applied on the qubits.
The analysis showed that if the entangling unitary has a higher KC number, more qubits need to be initially disentangled, which is the first instance to our knowledge in which a clear relation between entanglement conversion and the KC number has been drawn.
Our result shows that the difference between global operations and LOCC appears not only in entangling process, but also in \textit{dis}entangling process in the multipartite setting.
We also proved how the KC number changes when two or more two-qubit unitaries are applied sequentially.
Finally, LOCC partial inversion is an example of entanglement conversion in which the entangling unitary could not be characterized by its operator Schmidt number, since unitaries with KC number 2 and 3 cannot be distinguished by the operator Schmidt number.
It is an open question why such difference appears and whether it applies to other tasks as well.

\ack
This work is supported by the Project for Developing Innovation Systems of MEXT, Japan, the Global COE Program of MEXT Japan, and JSPS KAKENHI (Grant No. 23540463, and No. 23240001).

\end{document}